%% file: graph-balancing.tex
\documentclass[a4paper,UKenglish,cleveref, autoref, thm-restate,authorcolumns]{article}
\usepackage[a4paper,margin=1in]{geometry}
\usepackage{amsmath,amsthm,amssymb}
\usepackage{graphicx}
\usepackage{tikz}
\usepackage{hyperref}

\newtheorem{theorem}{Theorem}
\newtheorem{corollary}[theorem]{Corollary}
\newtheorem{lemma}[theorem]{Lemma}
\newtheorem{proposition}[theorem]{Proposition}

\newtheorem{example}{Example}

\theoremstyle{definition}
\newtheorem{definition}[theorem]{Definition}

\theoremstyle{remark}

\newcommand*\circled[1]{\tikz[baseline=(char.base)]{
            \node[shape=circle,draw,inner sep=2pt] (char) {#1};}}

\newcounter{note}[section]

\title{Truthful allocation in graphs and hypergraphs}

\author{George Christodoulou\thanks{Department of Computer Science,
    University of Liverpool, UK.  Email:
    \texttt{gchristo@liverpool.ac.uk} } \and Elias Koutsoupias \thanks{Department of Computer Science, University
    of Oxford, UK. Email: \texttt{elias@cs.ox.ac.uk}} \and Annam{\'a}ria
  Kov{\'a}cs\thanks{Department of Informatics, Goethe University,
    Frankfurt M., Germany. Email: \texttt{panni@cs.uni-frankfurt.de}}} 

\date{}

\DeclareMathOperator*{\argmin}{arg\,min}

\begin{document}

\maketitle

\begin{abstract}
  We study truthful mechanisms for allocation problems in graphs, both for the
  minimization (i.e., scheduling) and maximization (i.e., auctions) setting. The
  minimization problem is a special case of the well-studied unrelated machines
  scheduling problem, in which every given task can be executed only by two
  pre-specified machines in the case of graphs or a given subset of machines in the
  case of hypergraphs. This corresponds to a multigraph whose nodes are the machines
  and its hyperedges are the tasks. This class of problems belongs to
  multidimensional mechanism design, for which there are no known general mechanisms
  other than the VCG and its generalization to affine minimizers.  We propose a new
  class of mechanisms that are truthful and have significantly better performance
  than affine minimizers in many settings. Specifically, we provide upper and lower
  bounds for truthful mechanisms for general multigraphs, as well as special classes
  of graphs such as stars, trees, planar graphs, $k$-degenerate graphs, and graphs of
  a given treewidth. We also consider the objective of minimizing or maximizing the
  $L^p$-norm of the values of the players, a generalization of the makespan
  minimization that corresponds to $p=\infty$, and extend the results to any $p>0$.
\end{abstract}

\input{intro}

\input{contribution}

\input{related}

\input{preliminaries}

\input{graphs}

\input{lower-bounds}

\input{extensions}

\input{objectives}

\input{hybrid}

\bibliographystyle{plain} 
\bibliography{biblio}

\end{document}

%% file: intro.tex
\section{Introduction}

This work belongs to the area of mechanism design, one of the most
researched branches of Game Theory and Microeconomics with numerous
applications in environments where a protocol of conduct of selfish
participants is required. The goal is to design an algorithm, called
mechanism, which is robust under selfish behavior and that produces a
social outcome with a certain guaranteed quality. The mechanism
solicits the preferences of the participants over the outcomes, in
forms of bids, and then selects one of the outcomes. The challenge
stems from the fact that the real preferences of the participants are
private, and the participants care only about maximizing their private
utilities and hence they will lie if a false report is profitable. A
{\em truthful} mechanism provides incentives such that a truthful bid
is the best action for each participant.

Despite the importance of the problem the only general positive result
for multi-dimensional domains is the celebrated Vickrey-Clarke-Groves (VCG)
mechanism~\cite{Vic61,Cla71,Gro73} and its affine
extensions, known as affine maximizers. %

In their seminal paper on algorithmic mechanism design, Nisan and Ronen~\cite{NR01}
proposed the scheduling problem on unrelated machines as a central problem to
understand the algorithmic aspects of mechanism design. The objective is to
incentivize $n$ machines to execute $m$ tasks, so that the maximum completion time of
the machines, i.e. the makespan, is minimized. Scheduling, a problem that has been
extensively studied from the classical algorithmic perspective, proved to be the
perfect ground to study the limitations that truthfulness imposes on algorithm
design.

Nisan and Ronen applied the VCG mechanism, the most successful generic machinery in
mechanism design, which truthfully implements the outcome that maximizes the social
welfare.  In the case of scheduling, the allocation of the VCG is the greedy
allocation in which each task is assigned to the machine with minimum processing
time. This mechanism is truthful, but has a poor approximation ratio of $n$ for the makespan. They
conjectured that this is the best guarantee that can be achieved by any deterministic
(polynomial-time or not) truthful mechanism and this conjecture, known as the
Nisan-Ronen conjecture, is widely perceived as the holy grail in algorithmic mechanism
design.

An interesting special case of the scheduling problem, which is well-understood, is
the single-dimensional mechanism design in which the values of each player are linear
expressions of a single parameter. The principal representative is the problem of
scheduling {\em related} machines, where the cost of each machine can be expressed
via a single parameter, its {\em speed}. This was first studied by Archer and
Tardos~\cite{AT01}, who showed that in contrast to the unrelated machines version, an
algorithm that minimizes the makespan can be truthfully implemented --- albeit in
exponential time. It was subsequently shown that truthfulness has essentially no
impact on the computational complexity of the problem. Specifically, a randomized
truthful-in-expectation\footnote{This is one of the two main definitions of
  truthfulness for randomized mechanisms, where truth-telling maximizes the expected
  utility of each player.} PTAS was given in~\cite{DDDR11} and a deterministic PTAS
was given in~\cite{CK13}; a PTAS is the best possible algorithm even for the pure
algorithmic problem (unless $P=NP$).

%% file: contribution.tex
\subsection{Summary of Results}

In this work, we show how to combine these two main positive results
of VCG and single-dimensional mechanisms into a single mechanism,
which we call the {\em Hybrid Mechanism}. This new mechanism applies
to domains in which some players are multidimensional and some players
are single-dimensional.  A typical example is to schedule $m$ tasks,
such that task $i$ can only be executed by player 0 and player $i$. In
this case, player 0 is multidimensional and the other $m$ players are
single-dimensional. We call this the {\em star balancing}
problem. This is a multidimensional mechanism design problem for which
the VCG mechanism, as well as every other known mechanism, performs
very poorly. However, as we show in Section~\ref{sec:stars}, the Hybrid
Mechanism has approximation ratio 2, optimal among all truthful
mechanisms. We generalize the star balancing problem in three
directions: \emph{hyperstars}, \emph{graphs/multigraphs} and also to
objectives other than makespan minimization. 

\subparagraph{Hyperstars.}
In the hyperstar version, there are $k$ multidimensional
players/machines and every task can be executed by any one of these
$k$ players or by a task-specific single-dimensional
player. Specifically, there are $k$ different \emph{root players}
(players $1,2,\ldots ,k$ with bids $(r_{ij})_{k\times m}$) and each of
them are allowed to process all tasks. In addition, for each task
there is one \emph{leaf player}, which can process only this single
task (players $k+1,k+2,\ldots,k+m$ with bids $(\ell_1,\ell_2,\ldots
\ell_m)$). Note that the root players without the leaves form a
classic input for unrelated scheduling mechanisms with $k$ players and
$m$-tasks. We can now state the Hybrid Mechanism for this
case.

\begin{definition}[Hybrid Mechanism]
\label{def:hybrid}
  The Hybrid Mechanism minimizes 
$$\min_T\left\{\left(\min_{x^T}\sum_{i=1}^k\lambda_i r_i\cdot x^T_i\right)
  + g_{T}(\ell)\right\},$$ where the $\lambda_i$ can be arbitrary non-negative real
numbers and $(g_T)_{T\subseteq M}$ can be any functions that guarantee that the leaf
players are truthful\footnote{In  Section~\ref{sec:hybrid} we provide
  general definitions as well as necessary and sufficient conditions for truthfulness of the Hybrid Mechanism. As a prominent example think of  $g_T(\ell)=\max_{j\notin T} \ell_i.$ }.  The output of the mechanism is the subset of tasks $T$ that
are allocated to the multidimensional root players together with their allocation
$x^T$ (i.e. its characteristic matrix $(x^T_{ij})_{k\times m},$ with exactly one $1$ in each column $j\in T$). The remaining tasks, $M\setminus T$, are allocated to the leaf players.
\end{definition}

VCG fairs poorly, yielding %
approximation ratio $m$ in this domain, but the Hybrid Mechanism has
approximation ratio $k+1$, as we show in the next theorem.
\begin{theorem} \label{thm:hyperstars}
  For the hyperstar scheduling problem, the Hybrid Mechanism with
  $g_T(\ell)=\max_{j\notin T} \ell_i,$ and with $\lambda_i=1$, for every $i$, is $(k+1)$-approximate.
\end{theorem}

\subparagraph*{(Multi)Graphs.}
The other generalization of the star balancing problem to graphs and
multigraphs is the Unrelated Graph Balancing problem
(Section~\ref{sec:scheduling}).  This is a special case of
unrelated machines scheduling in which there is a (multi)graph whose
nodes represent the machines and whose edges represent tasks that can be
executed only by the incident nodes. For general graphs, all machines
are multiparameter, but we can still apply the Hybrid Mechanism, if we
first decompose the graph into stars and then apply the Hybrid
Mechanism to each one of them. The combined mechanism, which we call
the {\em Star-Cover Mechanism}, has surprisingly good approximation ratio
for certain classes of graphs --- ratio 4 for trees, 8 for planar
graphs, and $2k+2$ for $k$-degenerate graphs
(Corollary~\ref{cor:planar}). These results use as ingredient the
analysis of star graphs, in which the Hybrid Mechanism has
approximation ratio 2 %
(Section~\ref{sec:scheduling}).

\subparagraph*{Mechanisms for ${L^p}$-norm optimization.}
In Section~\ref{sec:lp-norm}, %
we consider the much more general objective of
minimizing or maximizing the $L^p$-norm of the values of the players, for $p>0$. The
scheduling problem is the special case of minimizing the $L^{\infty}$-norm. We show
that the Hybrid Mechanism performs very well for this much more general problem, and
in some cases it has the optimal approximation ratio among all truthful mechanisms.
This illustrates the applicability and usefulness of the Hybrid Mechanism in
applications with various domains and objectives. We emphasize that for all these
cases, even for stars, all known mechanisms such as the VCG and affine maximizers
have very poor performance.

\subparagraph{Relation to the Nisan-Ronen conjecture.}  Our results on
(multi)graphs show that this domain may provide an easier way to
attack the Nisan-Ronen conjecture. In a recent work~\cite{CKK20}, we
showed a $\Omega(\sqrt{n})$ lower bound for multistars with edge
multiplicity only 2, when the root player has submodular or supermodular
valuations. In contrast, our results in this work show that for
additive valuations, the Star-Cover Mechanism has approximation ratio
4 on the very same multigraphs. However, the Hybrid and the Star-Cover Mechanisms have high approximation for
multistars with high edge-multiplicity or for simple clique graphs. It
is natural to ask whether there are other, better mechanisms for these
cases. Recently we have proved a $\Omega(\sqrt{n})$ lower bound for
the former case, which is the first super-constant lower bound for the
Nisan-Ronen problem~\cite{CKK20b}, and we conjecture that the latter case admits
similarly a high, perhaps even linear, lower bound.

We remark that all previous lower bound proofs use inherently either
(multi)graphs \cite{ChrKouVid09,KV07,CKK20b} or, recently,
hypergraphs with hyperedges of small size~\cite{GiannakopoulosH20,
  DS20}. Our work provides new methodological tools to study these
objects, that can help to identify certain (hyper)graph structures as
good candidates for high lower bounds and to avoid those where low
upper bounds exist.  For example, the 2.755 lower bound construction
of \cite{GiannakopoulosH20} uses a hyperstar with $k=2$, for which the Hybrid
Mechanism achieves an upper bound of 3 (Thm~\ref{thm:hyperstars}).

All our lower bounds are information theoretic and hold independently of the
computational time of the mechanisms. Conversely, all upper bounds are polynomial
time algorithms when the star decomposition is given. We leave it open whether
computing an optimal star decomposition of a graph is in $P$, although it follows
from our results that it can be approximated with an additive term of $1$ in
polynomial time (actually in linear time).

%% file: related.tex
\subsection{Related Work}
The Nisan-Ronen conjecture~\cite{NR01} has become one of the central
problems in Algorithmic Game Theory, and despite intensive efforts it
remains open. The original paper showed that no truthful deterministic
mechanism can achieve an approximation ratio better than $2$ for two
machines, which was later improved to $2.41$~\cite{ChrKouVid09} for
three machines, and finally to $2.618$ \cite{KV07} which was the best
known bound for over a decade.  Recent progress improved this bound to
$2.755$~\cite{GiannakopoulosH20}, to $3$ \cite{DS20} and finally to
the first non-constant lower bound of $1+\sqrt{n-1}$~\cite{CKK20b}. The
best known upper bound is $n$~\cite{NR01}.

The purely algorithmic problem of makespan minimization on unrelated machines is one
of the most important scheduling problems. The seminal paper of Lenstra, Shmoys and
Tardos~\cite{lenstra1990approximation}, gave a $2$-approximation algorithm, and also
showed that it is NP-hard to approximate within a factor of $3/2$. Closing
this gap has remained open for 30 years, and is considered one of the most important
open questions in scheduling.

In this work we consider the design of truthful mechanisms for the {\em Unrelated
  Graph Balancing} problem, a special but quite rich case of the unrelated machines
problem, which was previously studied by Verschae and Wiese~\cite{VerschaeW14}, for
which each task can only be assigned to two machines. This can be formulated as a
graph problem, where given an undirected (multi)-graph $G=(V,E)$, each vertex
corresponds to a machine, and each edge corresponds to a task. %
The goal is to allocate (direct) each edge to
one of its nodes, in a way that minimizes the maximum (weighted) in-degree.

The special case of this problem where each direction of an edge corresponds to
the same processing time $t(e)$ is known as Graph Balancing, and was
introduced by Ebenlendr, Krc\'al, and Sgall~\cite{EKS14} who showed an
$1.75$-approximate algorithm and also demonstrated that the problem
retains the hardness of the unrelated machines problem, by showing
that it is NP-hard to approximate within a factor better than $3/2$.

\input{further}

%% file: further.tex
\subsubsection{Further Related Work}
\label{sec:further}

\noindent\textbf{Graph Balancing.}
As was already mentioned, for the pure graph balancing problem, the best
approximation ratio for classical polynomial time algorithms is $1.75$ by
\cite{EKS14}. Wang and Sitters~\cite{WangS16} showed a different LP-based algorithm
with a higher ratio of $11/6\approx 1.83$, while Huang and Ott~\cite{HuangO16}
designed a purely combinatorial approximation algorithm but with also a higher
guarantee of $1.857$.

Jansen and Rohwedder~\cite{JansenR19} studied the so-called {\em configuration LP}
which was introduced by Bansal and Sviridenko~\cite{BansalS06}. They showed that it
has an integrality gap of at most 1.749 breaking the 1.75 barrier of the integrality
gaps of the previous LP formulations. This leaves open the possibility of using this
LP to produce an approximation algorithm with a ratio better than 1.75.

Verschae and Wiese~\cite{VerschaeW14} studied the {\em unrelated} version of graph
balancing (whose strategic variant we consider in this paper) and showed that the
integrality gap of the configuration LP is equal to $2$, which is much higher
comparing to graph balancing. They also showed a 2-approximation algorithm for the
problem of maximizing the minimum load, which is the best possible unless P=NP.

The problem has been studied for various special graph classes. For the case of
simple graphs (also known as Graph Orientation), Asahiro et
al~\cite{AsahiroJMO11-trees} showed that the problem is in P for the case of trees,
while Asahiro, Miyano and Ono~\cite{AsahiroMO11-planar} showed that it becomes
strongly NP-hard for planar and bipartite graphs. Finally, Lee, Leung and
Pinedo~\cite{LeeLP09a} concluded the case of trees in the case of multiple edges,
showing an FPTAS which is the best possible, given that the problem in multi-graphs
is immediately NP-hard even for the simple case of two vertices (due to reduction from Subset Sum).

\noindent\textbf{Truthful Scheduling.} The lack of progress in the original
unrelated machine problem led to the study of special cases where progress has been made. Ashlagi et
al.\cite{ADL09}, resolved a restricted version of the Nisan-Ronen conjecture, for the
special but natural class of {\em anonymous} mechanisms. Lavi and
Swamy~\cite{LaviS09} studied a restricted input domain which however retains the
multi-dimensional flavour of the setting. They considered inputs with only two
possible values ``low'' and ``high'', that are publicly known to the designer. For
this case they showed an elegant deterministic mechanism with an approximation factor
of 2. They also showed that even for this setting achieving the optimal makespan is
not possible under truthfulness, and provided a lower bound of
$11/10$. Yu~\cite{Yu09} extended the results for a range of values, and Auletta et
al.~\cite{Auletta0P15} studied multi-dimensional domains where the private
information of the machines is a single bit.

Randomization has led to mildly improved guarantees. There are two
extensions of truthfulness for randomized mechanisms;{\em universal
  truthfulness} if the mechanism is described as a probability
distribution over deterministic truthful mechanisms, and {\em
  truthfulness-in-expectation}, if in expectation no player can
benefit by lying. The former notion was first considered in
\cite{NR01} for two machines, it was later extended to $n$ machines by
Mu'alem and Schapira~\cite{MualemS18} and finally Lu and
Yu~\cite{LuYu08} showed a $0.837n$-approximate mechanism, which is
currently the best known. Lu and Yu~\cite{LuY08a} showed a
truthful-in-expectation mechanism with an approximation guarantee of
$(m+5)/2$.
Mu'alem and Schapira~\cite{MualemS18}, showed a lower bound of
$2-1/m$, for both notions of randomization. Christodoulou,
Koutsoupias and Kov{\'a}cs~\cite{CKK10} extended this lower bound for
fractional mechanisms, where each task can be split to multiple
machines, and they also showed a fractional mechanism with a guarantee
of $(m+1)/2$.
The special case of two machines~\cite{Lu09, LuY08a} is still
unresolved; currently, the best upper bound is $1.587$ due to
Chen, Du, and Zuluaga~\cite{ChenDZ15}.

The case of {\em related} machines is well understood. It falls into
the so-called {\em single-dimensional} mechanism design in which the
valuations of a player are linear expressions of a single parameter.
In this case, the cost of each machine is expressed via a single
parameter, its {\em (inverse) speed} multiplied by the workload
allocated to the machine, instead of an $m$-valued vector, as it is
the case for the unrelated machines and the Graph Balancing setting.
Archer and Tardos~\cite{AT01} showed that, in contrast to the
unrelated machines version, the optimal makespan can be achieved by an
(exponential-time) truthful algorithm, while \cite{CK13} gave a
deterministic truthful PTAS which is the best possible even for the
pure algorithmic problem (unless P=NP).

Truthful implementation of other objectives was considered by Mu'alem
and Schapi\-ra~\cite{MualemS18} for multi-dimensional problems and by
Epstein, Levin and van Stee~\cite{EpsteinLS13} for single-dimensional
ones. Leucci, Mamageishvili and Penna~\cite{LeucciMP18} demonstrated
high lower bounds for other min-max objectives on some combinatorial
optimization problems on graphs, showing essentially that VCG is the
best mechanism for these problems. Minooei and Swamy~\cite{MS12}
considered a multi-dimensional vertex cover problem, and approached it
by decomposition into single parameter problems.

The Bayesian setting, where the players costs are drawn from a
probability distribution has also been studied. Daskalakis and
Weinberg~\cite{DaskalakisW15} showed a mechanism that is at most a
factor of 2 from the {\em optimal truthful mechanism}, but not with
respect to the optimal makespan. 
Chawla et al.~\cite{ChawlaHMS13} provided bounds of prior-independent
mechanisms (where the input distribution is unknown to the mechanism),
while Giannakopoulos and Kyropoulou~\cite{GiannakopoulosK17} showed
that the VCG mechanism achieves a factor of $O( \log n/\log \log n )$
under some distributional and symmetry assumptions.

Recently Christodoulou, Koutsoupias, and Kov{\'a}cs~\cite{CKK20} showed
a lower bound of $\sqrt{n-1}$ for all deterministic truthful
mechanisms, when the cost of processing a subset of tasks is given by a submodular
(or supermodular) set function, instead of an additive function which
is assumed in the standard scheduling setting.

%% file: preliminaries.tex
\section{Preliminaries}
\label{sec:preliminaries}

\noindent\textbf{Scheduling.} In the classical \emph{unrelated machines scheduling} there is a set $N$ of $n$
machines and a set $M$ of $m$ tasks that need to be scheduled on the machines. The
input is given by a nonnegative matrix $t=(t_{ij})_{n\times m}:$ machine
$i$ needs time $t_{ij}\in \mathbb R_{\geq 0}$ to process task $j,$ and her
\emph{costs} are additive, i.e., the processing time for machine $i$ for a set of
tasks $X_i\subset M$ is $t_i(X_i):=\sum_{j\in X_i} t_{ij}.$ The objective is to
minimize the makespan (min-max objective).  An allocation to all machines
$X=(X_1,X_2,\ldots ,X_n),$ (which is a partition of $M$) can also be denoted by the
characteristic matrix $x=(x_{ij})$ where $x_{ij}=1$ if $j\in X_i,$ and $x_{ij}=0$
otherwise. %

The current work essentially considers a special case of unrelated scheduling, in which
every task can be processed by two designated machines.  The tasks can thus be
modelled by the edges of a graph, and the associated problem is also known as
\emph{Unrelated Graph Balancing}.  More formally, in the Unrelated Graph Balancing
problem, there is a given undirected graph $G=(V,E);$ the vertices correspond to a
set of machines $N=V$ and the edges to a set of tasks $M=E.$ For each edge $e\in E$
only its two incident vertices can process the job $e,$ and they have in general
different processing times $t_i(e),$ and $t_{i'}(e).$ The goal is to assign a direction
to each edge $e=(i,i')$ (allocate the corresponding task) of the graph, to one of the
incident vertices (machines). The \emph{completion time} of each vertex $i$ is then
the total processing time of the jobs $X_i$ assigned to it
$t_i(X_i)=\sum_{e\in X_i}t_{i}(e)$. The objective is to find an allocation that
minimizes the {\em makespan}, i.e. the maximum completion time over all vertices.

\noindent\textbf{Mechanism design setting.} We assume that each machine $i\in N$ is controlled by a selfish agent
that is reluctant to process the tasks and the cost function $t_i$ is
private information (also called the {\em type} of
agent $i$). A \emph{mechanism} asks the agents to report \emph{(bid)}
their types $t_i,$ and based on the collected bids it allocates the
jobs, and gives payments to the agents. A player may report a false
cost function $b_i\neq t_i$, if this serves her interests.

Formally, a mechanism $(X,P)$ consists of two parts:
\begin{description}
\item[An allocation algorithm:] The allocation algorithm $X$ allocates the tasks to the machines depending
  on the players' bids $b=(b_1,\ldots ,b_n)$. We denote by $X_i(b)$ the subset of tasks assigned to machine $i$ in the bid profile $b.$

\item[A payment scheme:] The payment scheme $P=(P_1,\ldots,P_n)$ determines the payments also
  depending on the bid values $b.$ The functions $P_1,\ldots,P_n$ stand
  for the payments that the mechanism hands to each agent.
\end{description}

The {\em utility} $u_i$ of a player $i$ is the payment that she gets
minus the {\em actual} time that she needs to process the set of tasks
assigned to her, $u_i(b)=P_i(b)-t_i(X_i(b))$. We are interested in
\emph{truthful} mechanisms. A mechanism is truthful, if for every
player, reporting his true type is a \emph{dominant
  strategy}. Formally,

$$u_i(t_i,b_{-i})\geq u_i(t'_i,b_{-i}),\qquad \forall i\in N,\;\; t_i,t'_i\in \mathbb R_{\geq 0}^m, \;\; b_{-i}\in \mathbb R_{\geq 0}^{(n-1)\times m},$$
where $b_{-i}$ denotes the reported bidvectors of all players disregarding
$i.$

We are looking for \emph{truthful} mechanisms with \emph{low approximation ratio} of
the allocation algorithm for the makespan  irrespective of the running time to
compute $X$ and $P.$ In other words, our lower bounds are information-theoretic and
do not take into account computational issues.

A useful characterization of truthful mechanisms in terms of the
following monotonicity condition, helps us to get rid of the payments
and focus on the properties of the allocation algorithm.

\begin{definition}[Weak Monotonicity] \label{def:wmon}
An allocation algorithm $X$ is called {\em
  weakly monotone (WMON)} if it satisfies the following property: for
every two inputs $t=(t_i,t_{-i})$ and $t'=(t'_i,t_{-i})$, the
associated allocations $X$ and $X'$ satisfy $t_i(X_i)-t_i(X'_i)\leq
t'_i(X_i)-t'_i(X'_i).$
\end{definition}

It is well known that the allocation function of every truthful
mechanism is WMON~\cite{BCR+06}, and also that this is a sufficient
condition for truthfulness in convex domains~\cite{SY05}.

The following lemma was essentially shown in \cite{NR01} and has been
a useful tool to show lower bounds for truthful mechanisms for several
variants (see for example \cite{ChrKouVid09,MS07}). %

\begin{lemma}\label{lemma:tool}
  Let $t$ be a bid vector, and let $S=X_i(t)$ be the subset assigned to
  player $i$ by a weakly monotone allocation $X.$ Let $t'=(t'_i,t_{-i})$ be a bid vector
  such that only the bid of machine $i$ has changed and in such a way
  that for every task in $S$ it has decreased (i.e., $t'_{ij}<
  t_{ij}, j\in S$) and for every other task it has increased
  (i.e., $t'_{ij}> t_{ij}, j\in M\setminus S$).  Then the mechanism does not
  change the allocation to machine $i$, i.e., $X_i(t')=X_i(t)=S$.
\end{lemma}

In general, when the values of a machine change, the allocation of the other machines
may change, this issue being the pivotal difficulty of truthful unrelated scheduling.
Allocation algorithms that ``promise'' not to change the allocation of other machines
as long as changing (only) $t_i$ does not affect the set $X_i,$ are less
problematic. These allocation rules are called \emph{local} in \cite{NR01}, where it
is shown that local truthful mechanisms cannot have a better than $n$ approximation.

\begin{definition}[Local mechanisms]\label{def:local} A mechanism is {\em local} if for every $i\in N$,
for every $t_{-i}$, and $t_i,t'_i$ for which
$X_i(t_i,t_{-i})=X_i(t'_i,t_{-i})$  also holds that $X_j(t_i,t_{-i})=X_j(t'_i,t_{-i})\quad (\forall j\in N).$  
\end{definition}

There are several special classes
of mechanisms that satisfy this property, perhaps the most prominent
one is the class of \emph{affine minimizers} (see, e.g., \cite{ChrKouVid09}).

%% file: graphs.tex
\section{Graph Balancing}
\label{sec:scheduling}

In this section we focus on the (Unrelated) Graph Balancing problem, which is a
special case of makespan minimization of scheduling unrelated machines. The Graph
Balancing is a multi-parameter mechanism design problem that retains most of the
difficulty of the Nisan-Ronen conjecture, yet has certain features that make it more
amenable.

One of the difficulties in dealing with truthful mechanisms is that while truthfulness is a
local property (i.e., independent truthfulness conditions, one per player), the
allocation algorithm is a global function (that involves all players). Local algorithms attempt
to reconcile this tension by insisting that the allocation is also ``local'', but they take
this notion too far. The results of this work show that locality in mechanisms is
very restrictive in some domains, where the Hybrid Mechanism outperforms every local
mechanism.

The Graph Balancing problem is more amenable than the general scheduling problem
because it exhibits another kind of locality, \emph{domain locality}: when a machine
does not get a task, we know which machines gets it. Yet, this locality is not very
restrictive and the problem retains most of its original difficulty.

In this section, we take advantage of domain locality to obtain an optimal mechanism
for stars. It turns out that this mechanism, the Hybrid Mechanism, is a special case
of a more general mechanism (see Section~\ref{sec:hybrid}). But since the Hybrid Mechanism does not apply to general
graphs, here we also propose the Star-Cover mechanism for general graphs: decompose the
graph into stars and apply the Hybrid Mechanism independently to each star. In this
way, we obtain a 4-approximation algorithm for trees and similar positive results for
other types of graphs.

Makespan minimization is the special case, when $p=\infty$, of minimizing the
$L^p$-norm of the values of the players. Other special cases of the $L^p$-norm
optimization is the case $p=1$, which corresponds to welfare maximization, and the
case $p=0$, which is related to Nash Social Welfare~\cite{Cole_2018}. We deal with
this more general problem in another section (Section~\ref{sec:lp-norm}). Most of the
results and proofs of this section generalize to any $p\geq 1$. %

\subsection{Stars and the Hybrid Mechanism}
\label{sec:stars}

In this subsection, we focus on star graphs, where there are $n=m+1$ players and $m$
tasks. Player $0$ is the root of the star, and has processing times given by a vector
$r=(r_1,r_2,\ldots r_m).$ We also refer to this player as the \emph{root player} or
\emph{$r$-player}. For given bids $r$ of the root player, and task set $T\subseteq M$
we use the short notation $r(T)=\sum_{j\in T} r_j$.

There are also $m$ \emph{leaf-players}, 
one for each leaf of the star with processing times $\ell=(\ell_1,\ldots,
\ell_m)$ respectively.  Each task $j$ can only be assigned to two players;
either to the root, with processing time $r_j,$ or to the leaf with
processing time $\ell_j$.

As usual, we denote by $r_{-i}$ the vector of bids of the root player except for the
bid for task $i,$ and similarly $\ell_{-i}$ denotes the bids of all leaf-players,
except for player $i.$ The vector of all input bids is given by $t=(r,\ell).$

As we show later in the Lower Bound section (Section~\ref{sec:lower-bounds}), all
previously known mechanisms for the Unrelated Graph Balancing problem, e.g. affine
minimizers and task independent mechanisms, have approximation ratio at least
$\sqrt{n-1}$ for graphs, even for stars.

In contrast, we now show that the Hybrid Mechanism has constant approximation ratio
for stars.

\begin{definition}[Hybrid Mechanism for Graph Balancing]\label{def:maxmech}
  Consider an instance of the Unrelated Graph Balancing problem on a star of $n$
  nodes and set of tasks $M$. Let
  \begin{align} \label{eq:max1}
    S\in \argmin_{T\subseteq M}\{r(T)+\max_{i\not \in  T} \ell_i\}.
  \end{align}
  The mechanism assigns a set of tasks $S$ to the root and the remaining tasks to
  leaves. Ties are broken in a deterministic way (e.g., lexicographically).
\end{definition}

Figure~\ref{fig:max-mechanism} shows the partition of the space of the root player
induced by the Hybrid Mechanism for a star of two leaves. 

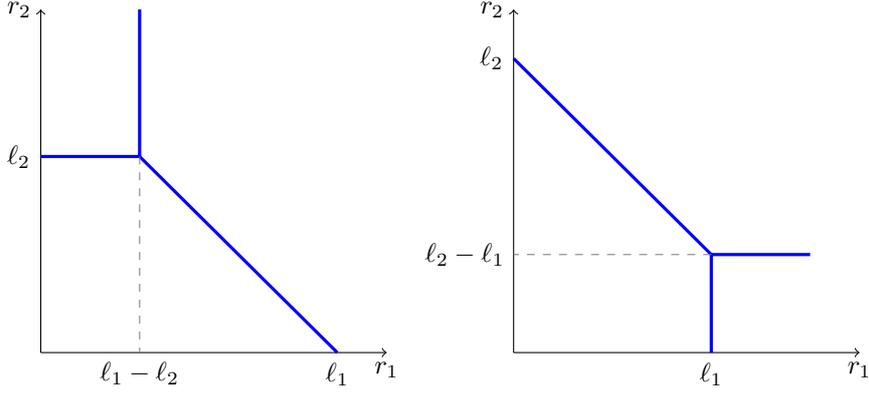
\begin{figure}
\centering
\begin{tikzpicture}[scale=0.65]

\draw[->] (0,0) -- (7,0) node[anchor=north] {$r_1$};
\draw[->] (0,0) -- (0,7) node[anchor=east] {$r_2$};

\draw[very thick, blue] (0, 4) node[anchor=east,black] {$\ell_2$} -- (2, 4) -- (2, 7);
\draw[very thick, blue ] (6, 0) node[anchor=north,black] {$\ell_1$}-- (2, 4);

\draw[dashed, gray] (2,4) -- (2,0) node[anchor=north,black] {$\ell_1-\ell_2$};

\end{tikzpicture}
\begin{tikzpicture}[scale=0.65]

\draw[->] (0,0) -- (7,0) node[anchor=north] {$r_1$};
\draw[->] (0,0) -- (0,7) node[anchor=east] {$r_2$};

\draw[very thick, blue] (4, 0) node[anchor=north,black] {$\ell_1$} -- (4, 2) -- (6, 2);
\draw[very thick, blue ] (0, 6) node[anchor=east,black] {$\ell_2$}-- (4, 2);

\draw[dashed, gray] (4,2) -- (0,2) node[anchor=east,black] {$\ell_2-\ell_1$};

\end{tikzpicture}

\caption{An instance of the Hybrid Mechanism, for the star of $m=2$ leaves. It shows
  the partition of bid-space of the root player induced by the allocation of the Hybrid
  Mechanism when $\ell_1 \geq \ell_2$ (left) and when $\ell_2 \geq \ell_1$ (right).
  In the left case, the root gets both tasks in the area near $(0,0)$, it gets only
  task $1$ when $r_1\leq \ell_1-\ell_2$ and $r_2\geq \ell_2$, and it gets neither
  task otherwise. Note that, in contrast to VCG, for every vector of fixed values for
  the leaves, only three allocations are possible.}
\label{fig:max-mechanism}
\end{figure}

The argmin expression that defines the Hybrid Mechanism and a corresponding expression
that defines the VCG mechanism are similar: in the definition of VCG, instead of
$\max_{i\not \in T} \ell_i$, we have $\sum_{i\not \in T} \ell_i$. It is a happy
coincidence that replacing the operator sum with max preserves the truthfulness of
the mechanism, a fact that rarely holds.

\begin{lemma}
  The Hybrid Mechanism for Graph Balancing on stars is truthful and has approximation
  ratio 2.
\end{lemma}
\begin{proof}
  The root player has no incentive to lie since $-\max_{i\not \in T} \ell_i$ can be
  interpreted as its payments. The reason that leaf players have no incentive to lie
  comes essentially from the fact that the expression in \eqref{eq:max1} is monotone
  in $\ell_i$ (see Section~\ref{sec:hybrid}, for a more rigorous and extensive treatment of the
  truthfulness of the general Hybrid Mechanism).

  Let $S^*=\argmin_{T\subseteq M}\max\{r(T),\, \max_{i\not \in T} \ell_i\}$ be the subset
  assigned to the root in an optimal allocation,
  $OPT$ be the optimal
  makespan, and $ALG$ be the makespan achieved by the Hybrid Mechanism. Then we have
\begin{align*}
 ALG  \leq \min_{T\subseteq M}\{r(T)+\max_{i\not \in  T} \ell_i\} \leq r(S^*)+\max_{i\not \in  S^*} \ell_i\leq  2\max \{r(S^*),\,\max_{i\not \in  S^*} \ell_i \}= 2 OPT. 
\end{align*}

\end{proof}

\subsection{Upper bound for general graphs and multigraphs}
\label{sec:gener-graphs-mult}

We now turn our attention to positive (upper bound) results for general graphs and
multigraphs. We will need a few definitions first.

\begin{definition}[Star decomposition]
  A \emph{star decomposition} of a (multi)graph $G(V,E)$ is a partition
  $T=\{T_1,\ldots,T_k\}$ of its edges into stars (see
  Figure~\ref{fig:star-decomposition} for an example). Let $V(T_i)$ denote the vertex set of the star spanned by $T_i.$ The star contention number of
  a star decomposition is the maximum number of stars that include a node either as a
  root or as a leaf: $c(T)=\max_{v\in V}|\{i\,:\, v\in V(T_i), i=1,\ldots,k\}|$. The
  \emph{star contention number of a (multi)graph} is the minimum star contention number
  among all its star decompositions.
\end{definition}

In an optimal star decomposition of a graph (but not multigraph), we can assume that
every node is the root of at most one star, otherwise we can merge stars with common
root without changing the star contention number.

A related notion to star decomposition that has been studied
extensively is the notion of edge orientation of a
multigraph (or of load balancing when we consider multigraphs).

\begin{definition}[Edge orientation number]
  Define the orientation number of a given orientation of the edges of a multigraph
  $G$, as its maximum in-degree. The \emph{edge orientation number} $o(G)$ of a multigraph
  $G$ is the minimum orientation number among all its possible orientations.
\end{definition}

Indeed the two notions are closely related: every star decomposition corresponds to a
graph orientation by orienting the edges in all stars from roots to leaves, and vice
versa a graph orientation gives rise to a star decomposition in which every node with
its outgoing edges defines a star. Given that in an optimal star decomposition of a
graph, each node is the root of at most one star, we get that for \emph{every graph}
$G$:
\begin{align*}
  o(G)\leq c(G) \leq o(G)+1.
\end{align*}
This relation for \emph{multigraphs} is similar only that in the right hand side we add the
maximum edge multiplicity $w$ instead of $1$, i.e., $o(G)\leq c(G) \leq o(G)+w$.

The following definition utilizes the Hybrid Mechanism on stars to obtain a general
mechanism for arbitrary graphs (and multigraphs). 
 
\begin{definition}[Star-Cover Mechanism] \label{def:star-cover}
  Let $G=(V,E)$ be a multigraph and let $T=\{T_1, \ldots, T_k\}$ be a fixed star
  decomposition. The Star-Cover mechanism runs the Hybrid Mechanism on every star of
  $T$ independently. That is, if $S_{i,h}$ is the subset of tasks allocated to a
  player $i$ by the Hybrid Mechanism when applied to a star $T_h$, the set of tasks
  allocated to player $i$ is $S_i=\cup_{h = 1}^k S_{i,h}$.
\end{definition}

We can now state and prove the general positive theorem of this section.
\begin{theorem}
  The Star-Cover mechanism for a given multigraph $G$ that uses the Hybrid Mechanism on
  every star of a fixed star decomposition $T=\{T_1,\ldots,T_k\}$ is truthful and has
  an approximation ratio at most $2c(T)$.
\end{theorem}
\begin{proof}
  Fix some player $i$ and let $S_{i,h}$ be the subset of tasks allocated to player
  $i$ by the Star Mechanism when applied to a star $T_h$,
  $h=1,\ldots,k$. Truthfulness is an immediate consequence of the following two
  observations. First, since the fixed star decomposition is independent of player
  $i$'s processing times, player $i$ cannot affect it by lying. Second, $S_{i,h}$ is
  independent of player $i$'s processing times $t_{i}(e)$ for all edges
  $e\not\in T_h,$ therefore player $i$ cannot alter the assignment on $T_h$ by
  changing its values outside $T_h$.

  To see the approximation guarantee, let $OPT$, $OPT(T_h)$ be the optimal makespan
  on $G$ and $T_h$ respectively, and let $ALG$ and $ALG(T_h)$ be the makespan
  achieved by the Star-Cover mechanism on $G$ and $T_h.$
  \begin{align*}
    ALG  \leq  \max_{h=1,\ldots,k} c(T)\cdot ALG(T_h)\leq \max_{h=1,\ldots,k}
    c(T)\cdot 2 OPT(T_h)\leq 2 c(T)\cdot OPT.
  \end{align*}
\end{proof}

Due to the close connection between star decompositions and edge orientations in
graphs, we get
\begin{corollary}
  The approximation ratio for graphs with edge orientation number $o(G)$ is at most
  $2o(G)+2$.
\end{corollary}

In the sequel, we consider particular bounds for certain classes of graphs.  It
is known that the edge orientation number of a given graph can be computed in
polynomial time~\cite{AsahiroJMO11-trees}. In fact, by an application of the
max-flow-min-cut theorem it can be shown that $o(G)\leq \gamma$ iff for every
subgraph $H$ of $G$ it holds that $|E(H)|\leq \gamma|V(H)|.$ Since this equivalent
condition\footnote{This characterization of the orientation number $o(G)$
  implies that a truthful mechanism with constant approximation ratio exists for
  any minor-closed class of graphs, because for every class of graphs with
  forbidden minors, there exists some constant $\gamma$ that satisfies the
  property (see Theorems 7.2.3, 7.2.4 and Lemma 12.6.1. in \cite{Dies12}). We are grateful to an anonymous referee for pointing this out.}
holds for planar graphs with $\gamma=3,$ we immediately obtain:

\begin{theorem} For every planar graph, there exists a truthful mechanism with approximation  ratio $8.$
\end{theorem}

A natural class of graphs fulfilling this property (with $\gamma=k$) is $k$-degenerate graphs. A graph $G(V,E)$ is
called \emph{$k$-degenerate}~\cite{erdHos1966chromatic} (or \emph{$k$-inductive}) if
there is an ordering $v_1,\ldots,v_n$ of its nodes such that the number of neighbors
of $v_i$ in $\{v_{i+1},\ldots,v_n\}$ is at most $k$. Many interesting classes of
graphs are $k$-degenerate for some small $k$. Besides planar graphs (with $k=5$), another example is given by $k$-trees~\cite{rose1974simple}: by definition, a $k$-tree is a
degenerate graph with an ordering such that every $v_i$ (except for the last $k$
nodes of the ordering) has exactly $k$ neighbors in $\{v_{i+1},\ldots,v_n\}$ and
these $k$ neighbors form a clique. Since graphs of treewidth $k$ are subgraphs of
$k$-trees~\cite{rose1974simple}, they are also $k$-degenerate. In particular, trees
are $1$-degenerate. We give here a direct proof and illustration of a star decomposition for $k$-degenerate graphs:

\begin{theorem} \label{thm:degenerate}
  For every $k$-degenerate graph, there is a truthful mechanism with approximation
  ratio $2k+2$.
\end{theorem}
\begin{proof}
  Consider a $k$-degenerate graph $G$. It suffices to show that it admits a star
  decomposition with contention number $k+1$. Let $v_1,\ldots,v_n$ be an inductive
  ordering of the nodes of $G$. We consider the star covering $\{T_2,\ldots,T_n\}$
  where $T_i$ is the star with root $v_i$ and leaves all its neighbors in
  $\{v_1,\ldots,v_{i-1}\}$. Note that stars are created in the opposite direction of
  the inductive order; see Figure~\ref{fig:star-decomposition} for an example. This
  star decomposition has contention number $k+1$ since every node belongs to at most one
  star as a root and to at most $k$ stars as a leaf.
  \begin{figure}
  \centering

\begin{tikzpicture}[scale=0.75]
  
\node[circle, draw, fill=blue!20] (1) at (0,1) {\tiny $1$};
\node[circle, draw, fill=blue!20] (2) at (0,2) {\tiny $2$};
\node[circle, draw, fill=blue!20] (3) at (0,3) {\tiny $3$};
\node[circle, draw, fill=blue!20] (4) at (0,4) {\tiny $4$};
\node[circle, draw, fill=blue!20] (5) at (0,5) {\tiny $5$};
\node[circle, draw, fill=blue!20] (6) at (0,6) {\tiny $6$};

\path
(6) edge (5)
    edge[bend left] (4)
    edge[bend right] (3)
    edge[bend left] (1)
(5) edge (4)
    edge[bend left] (3)
(4) edge[bend left] (2)
(3) edge (2)
(2) edge (1)
;

\node at (1.2,3.5) {$=$};

\node[circle, draw, fill=blue!20] (61) at (2,1) {\tiny $1$};
\node[circle, draw, fill=blue!20] (63) at (2,3) {\tiny $3$};
\node[circle, draw, fill=blue!20] (64) at (2,4) {\tiny $4$};
\node[circle, draw, fill=blue!20] (65) at (2,5) {\tiny $5$};
\node[circle, draw, fill=blue!20] (66) at (2,6) {\tiny $6$};

\path
(66) edge (65)
    edge[bend left] (64)
    edge[bend right] (63)
    edge[bend left] (61);t

\node at (3.2,3.5) {$+$};

\node[circle, draw, fill=blue!20] (53) at (4,3) {\tiny $3$};
\node[circle, draw, fill=blue!20] (54) at (4,4) {\tiny $4$};
\node[circle, draw, fill=blue!20] (55) at (4,5) {\tiny $5$};

\path
(55) edge (54)
    edge[bend left] (53);

\node at (5.2,3.5) {$+$};

\node[circle, draw, fill=blue!20] (42) at (6,2) {\tiny $2$};
\node[circle, draw, fill=blue!20] (44) at (6,4) {\tiny $4$};

\path
(44) edge[bend left] (42);

\node at (7.2,3.5) {$+$};

\node[circle, draw, fill=blue!20] (32) at (8,2) {\tiny $2$};
\node[circle, draw, fill=blue!20] (33) at (8,3) {\tiny $3$};

\path
(33) edge (32);

\node at (9.2,3.5) {$+$};

\node[circle, draw, fill=blue!20] (21) at (10,1) {\tiny $1$};
\node[circle, draw, fill=blue!20] (22) at (10,2) {\tiny $2$};

\path
(22) edge (21);

\end{tikzpicture}
  
  \caption{The star decomposition used in Theorem~\ref{thm:degenerate} of a
    2-degenerate graph. The inductive order is upwards, while the stars are
    ``pointing'' downwards.}
  \label{fig:star-decomposition}
\end{figure}
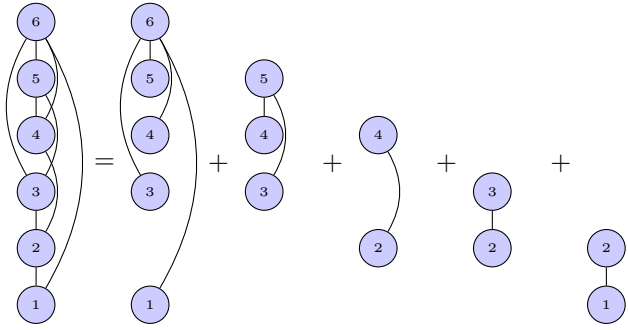
\end{proof}

\begin{corollary} \label{cor:planar}
  There exist truthful mechanisms with 
  approximation ratio at most 4 for trees, and generally of ratio at most $2k+2$ for graphs of treewidth $k.$
\end{corollary}

%% file: lower-bounds.tex
\subsection{Lower Bounds for Graph Balancing }
\label{sec:lower-bounds}

In this subsection, we show corresponding negative results for the positive results
of the previous subsection. We first observe that the natural candidate mechanisms
for the Graph Balancing problem have very poor performance, in stark contrast to the
Hybrid Mechanism.

\begin{theorem} \label{thm:local} All local mechanisms for stars, including VCG,
  affine minimizers and task-independent mechanisms, have approximation ratio at
  least $\sqrt{m}=\sqrt{n-1}.$ %
\end{theorem}

\begin{proof}
Consider the following input
$$t=\left(
\begin{array}{c c c c}
  \frac{1}{\sqrt{m}} & \frac{1}{\sqrt{m}} & \cdots & \frac{1}{\sqrt{m}}\\
  1 & \infty & \cdots &  \infty\\
  \infty & 1 & \cdots & \infty\\
  \infty  & \infty & \cdots & 1\\
\end{array}\right).$$

If, in the allocation of the mechanism, the root player takes all the tasks, then
this allocation has approximation $\sqrt{m}$, as the optimal allocation is to assign
the tasks to the leaves with makespan equal to 1. Otherwise, assume that (at least)
one of the tasks, is given to some other player, say w.l.o.g. task $1$ is given to
player $1.$ By a series of applications of Lemma~\ref{lemma:tool}, and by exploiting
the locality of the mechanism, we set the value of the owner of task $j$ to 0 for
every $j\neq 1.$

In particular, let $S$ be the set of tasks assigned to the root
player, and $M\setminus S$ be the tasks assigned to their respective leaf-player. %
Let
$t^1=(r',\ell_1,\ldots, \ell_m)$, with $r'$ defined as follows for some
arbitrarily small $\epsilon$.

$$r'_j=\left\{
\begin{array}{c c}
  0 & j\in S \\
  \frac{1}{\sqrt{m}}+\epsilon & \text{ otherwise. } \\
\end{array}\right.
$$

By applying Lemma~\ref{lemma:tool}, the root player receives again the set $S,$
 and therefore, the set $M\setminus S$ is assigned to the leaves. 
 We proceed by changing the bids of the leaf-players for the tasks in $M\setminus S$ to 0, i.e.,
 defining a sequence $t^j$ for $j\in M\setminus S$,
with $t^j=(r',\ell'_j=0,\ell^{j-1}_{-j})$ %

Again, by Lemma~\ref{lemma:tool} and by locality, we get that the allocation of the tasks remains the same for the leaf $j,$ \emph{and} 
for all the other players as well.

We end up with an instance $t'$ where player 1 still takes the first
task, while the rest of the tasks are assigned to a player with 0
processing time. For $t'$, the optimal makespan is $1/\sqrt{m}$, while
the mechanism achieves makespan equal to 1. 
We illustrate the case when $S=\emptyset,$ that is, the allocation  gives
all the tasks to the leaves of the star. %

$$t=\left(
\begin{array}{c c c c}
  \frac{1}{\sqrt{m}} & \frac{1}{\sqrt{m}} & \cdots & \frac{1}{\sqrt{m}}\\
  \circled{1} & \infty & \cdots &  \infty\\
  \infty & \circled{1} & \cdots & \infty\\
  \infty  & \infty & \cdots & \circled{1}\\
\end{array}\right) \rightarrow t'=\left(
\begin{array}{c c c c}
  \frac{1}{\sqrt{m}} & \frac{1}{\sqrt{m}} & \cdots & \frac{1}{\sqrt{m}}\\
  \circled{1} & \infty & \cdots &  \infty\\
  \infty & \circled{0} & \cdots & \infty\\
  \infty  & \infty & \cdots & \circled{0}\\
\end{array}\right)$$

\end{proof}

In the previous subsection, we showed that the Hybrid Mechanism outperforms all known
mechanisms and has approximation ratio at most 2. The next theorem shows that this
ratio is the best possible among all possible mechanisms for stars.

\begin{theorem} \label{thm:lower-bound-stars} There is no deterministic mechanism for
  stars that can achieve an approximation ratio better than 2.
\end{theorem}

This is a special case of a more general lower bound for the $L^p$-norm objective
(Theorem~\ref{thm:lower-bound-stars-lp}), but we give the proof here anyway, since it
will be an ingredient of the proof of the following theorem
(Theorem~\ref{thm:phi+1}).

\begin{proof}
  Let's assume that the mechanism takes an input where the processing
  time of the root player is $r_j =a^{j-1}$, for each task $j$, where $a>1$ is a parameter,
  and the
  processing time of the corresponding leaf player for task $j$ is
  $\ell_j=a^j$, as also shown in the following table.

$$t=\left(
\begin{array}{c c c c c}
  1 & a & \cdots & a^{m-2} & a^{m-1}\\
  a & \infty & \cdots &  \infty & \infty\\
  \infty & a^2 & \cdots & \infty & \infty\\
  \vdots  & \vdots & \vdots & \vdots & \vdots\\
  \infty  & \infty & \cdots & a^{m-1} & \infty\\
  \infty  & \infty & \cdots & \infty & a^m \\
\end{array}\right)$$

If the mechanism assigns all tasks to the root player, then the
makespan for this input is $(a^m-1)/(a-1)$, while the optimal makespan is
$a^{m-1}$, which yields a ratio of $(a^m-1)/((a-1)a^{m-1})$.

Otherwise, let $X$ be the nonempty set of tasks assigned to the leaf players. Let
$k$ be the task with the maximum index in $X$. Since it is processed by the leaf
player, its processing time is $a^k$. Now consider the input in which we change
the processing times of the root player to

$$r'_j=
\begin{cases}
  0 & j\not\in X \\
  r_j+\epsilon & \text{ otherwise} 
\end{cases}
$$
for some arbitrarily small $\epsilon>0$. By weak monotonicity
(Lemma~\ref{lemma:tool}), the set of tasks assigned to the root player remains the
same, and as a result the whole allocation stays the same. Therefore task $k$ is
still assigned to the leaf player $k$ and the makespan of the mechanism is at least
$a^k$. Notice that the optimum allocation for this input is $a^{k-1}+\epsilon$ which
yields an approximation ratio of $a$, as $\epsilon$ tends to $0$.

In conclusion, the approximation ratio is $\min\{(a^m-1)/((a-1)a^{m-1}),a\}$, for
every $a>1$. By choosing $a=2$, we see that the ratio is $2-1/2^{m-1}$, which shows
that for the class of stars no mechanism can have approximation ratio better than
$2$. For fixed $m$, the lower bound is slightly better than $2-1/2^{m-1}$, by
selecting $a$ to be the positive root of the equation $(a^m-1)/((a-1)a^{m-1})=a$.
\end{proof}

We now show how to extend the previous result to get a lower bound of
$1+\varphi\approx 2.618$ for trees, and thus for graphs. This matches the best lower bound for the Nisan-Ronen setting~\cite{KV07} that was known until the recent improvements~\cite{GiannakopoulosH20, DS20, CKK20b}, suggesting that studying the special case of scheduling in graphs may be useful in attacking the Nisan-Ronen conjecture.

\begin{figure}
  \centering
  \begin{minipage}[c]{0.4\linewidth}
  \begin{tikzpicture}[scale=0.8,%
  cnode/.style = {circle, draw, text centered, node distance=3cm, fill=blue!20},
  dummy/.style = {distance=6cm}]
  \node [cnode] {\small $0$}
    child { node [cnode] {\small $1$} }
    child { node [cnode] (2) {\small $2$} }
    child { node [dummy] {\small $\cdots$} }
    child { node [cnode] (n) {\small $m$} }
  ;

\end{tikzpicture}
\end{minipage}
\begin{minipage}[c]{0.4\linewidth}
  \begin{tikzpicture}[scale=0.8,%
  cnode/.style = {circle, draw, text centered, node distance=3cm, fill=blue!20},
  dummy/.style = {distance=3cm}]
  \node [cnode] {\small $\overline 0$}
    child { node [cnode] {\small $0$} 
      child { node [cnode] {\small $1$}
        child { node [cnode] {\small $\overline 1$} }}
      child { node [cnode] {\small $2$}
        child { node [cnode] {\small $\overline 2$} }}     
      child { node [dummy] {\small $\cdots$}}
      child { node [cnode] {\small $k$}
        child { node [cnode] {\small $\overline k$} }}
    }
  ;

\end{tikzpicture}
\end{minipage}

  \caption{A star with root $0$ and leaves $1,\ldots,m$ and its extension
    to a tree with dummy nodes}
  \label{fig:star}
\end{figure}
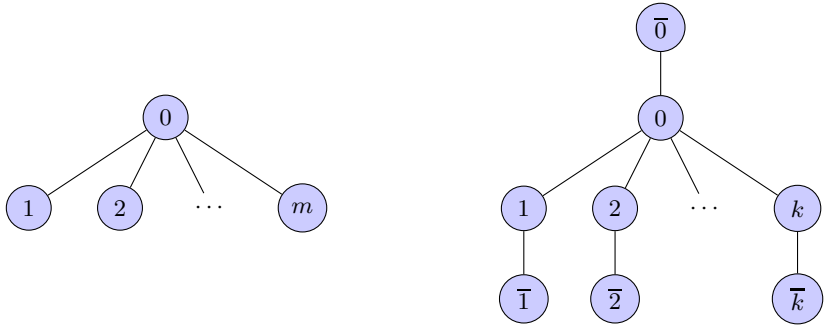

\begin{theorem} \label{thm:phi+1}
  No mechanism for trees can achieve approximation ratio $1+\varphi \approx 2.618$.
\end{theorem}
\begin{proof}
  The proof mimics the proof of Theorem~\ref{thm:lower-bound-stars} on the tree shown
  in Figure~\ref{fig:star}. The tree consists of a star with root $0$ and leaves
  $1,\ldots, k$ in which we add a new node $\overline v$ for each node $v$
  of the star and connect it to $v$. These new nodes (players), which we call dummy
  will not be assigned any task by any efficient mechanism since we set their processing
  times to an arbitrarily high value $H$.

  The processing times of the edges of the star are exactly the same as in the proof
  of Theorem~\ref{thm:lower-bound-stars}: $r_j=a^{j-1}$ and
  $\ell_j=a^{j}$, for some $a>1$. The processing times for all edges are given
  below:
  \begin{align*}
    r_j&=a^{j-1} & \ell_j&=a^j&j&=1,\ldots,k \\
    \overline r&=0 &   \overline \ell_j&=0 & &
  \end{align*}
  where $\overline r$ and $\overline \ell_j$ are the processing times of the star vertices of their respective dummy tasks. The dummy nodes themselves have a very large processing time $H\gg 1\,$ on these tasks.
  
  We consider two cases. In the first case, all tasks of the star are assigned to the
  root player $0.$ We then consider a new instance in which we slightly lower the
  processing time of the root on the tasks of the star (i.e.,
  $r_j=a^{j-1}-\epsilon$ for some $\epsilon>0$) and increase the processing
  time of its dummy task $\overline r=a^k$. By weak monotonicity
  (Lemma~\ref{lemma:tool}), the $r$-player  will take this task and all tasks of the star
  with a total processing time slightly less than
  $1+a+\ldots+a^k=(a^{k+1}-1)/(a-1)$. It is easy to see that the optimal allocation
  for this instance is $a^k$, and the approximation ratio $(a^{k+1}-1)/((a-1)a^k)$.

  In the second case, at least one task of the star is allocated to a leaf. Let $p$
  be the star task allocated to a leaf with the maximum index (that is, task $p$ of
  the star is allocated to leaf-player $p$ and tasks $p+1,\ldots,k$ are allocated to
  the root). We consider the instance in which we change the processing times of the
  root player as follows: all processing times of the tasks allocated to the root
  become $0$ and all processing times of the root player for the remaining tasks
  increase slightly. By weak monotonicity (Lemma~\ref{lemma:tool}), the $r$-player
  will still get the same set of tasks. We now create a new instance by increasing
  the processing time of the $p$-th dummy task: $\overline \ell_p=a^{p-1}$ and
  slightly decreasing the processing time of the leaf $p$ for its task in the star:
  $\ell_p=a^p-\epsilon$, for some $\epsilon>0$. Then again by weak monotonicity
  (Lemma~\ref{lemma:tool}), player $p$ will get these two tasks. Although the
  allocation of the other tasks may change, the cost for the mechanism is at least
  $a^p+a^{p-1}-\epsilon$, while the optimal allocation has cost $a^{p-1}$. Therefore,
  in this case the mechanism has approximation ratio $(a^p+a^{p-1})/a^{p-1}=a+1$, as
  $\epsilon\rightarrow 0$.

  In any case, the mechanism has approximation ratio
  $\min\{((a^{k+1}-1)/((a-1)a^k),a+1\}$. By selecting $a=\varphi$, we get a ratio
  at least $1+\varphi$ (as $k\rightarrow \infty$).
  \end{proof}

Closing the gap between the above lower bound 2.618 of Theorem~\ref{thm:phi+1} and
the upper bound 4 (Corollary~\ref{cor:planar}) for mechanisms for trees is a crisp
intriguing question.

%% file: extensions.tex
\section{Hypergraphs}
\label{sec:hypergraphs}

In this section we consider the generalization of some results of the
previous section from stars to hyperstars. A hyperstar is a hypergaph,
with a center consisting of $k$ different $r$-players (players
$1,2,\ldots ,k$ with bids $(r_{ij})_{k\times m}$), where each of them
can process all tasks, and a set of leaves (players
$k+1,k+2,\ldots,k+m$ with bids $(\ell_1,\ell_2,\ldots \ell_m)$), one
leaf per task. Each task can be allocated to any $r$-player and to its
associated leaf player. Note that the $r$-players without the leaves
form a classic input for unrelated scheduling mechanisms with $k$
players and $m$-tasks. The Hybrid Mechanism for this setting is
defined in Definition~\ref{def:hybrid}. Here we give the definition again for the special case when $g_T(\ell)=\max_{j\notin T} \ell_i,$ and prove that in this case the Hybrid Mechanism is $k+1$-approximative. In Section~\ref{sec:hybrid} we prove the truthfulness of the Hybrid Mechanism for stars and hyperstars for more general $g_T(\ell)$ functions, implying truthfulness in this special case (see Corollary~\ref{cor:sufficient}). 
 \begin{definition}[Hybrid (Max-)Mechanism for hyperstars]
   Let
 $$S\in \arg\min_T\bigg\{\bigg(\min_{x^T}\sum_{i=1}^k\lambda_i \,r_i\cdot x^T_i\bigg)  +\max_{j\notin T} \ell_i\bigg\},$$
 where the $\lambda_i$ are arbitrary non-negative real numbers, 
 the $x^T$ all possible characteristic matrices for allocations of tasks
 from $T$ and ties are broken by a fixed order over the
 allocations.  Assign $S$ to the root, and the rest of the items
 to the leaves.
 \end{definition}

The next theorem provides a general approximation ratio of the Hybrid Mechanism that
generalizes some results of the previous subsection.

\begin{theorem}[Theorem 2]%
  For the hyperstar scheduling problem, the Hybrid Mechanism with
  $g_T(\ell)=\max_{j\notin T} \ell_i,$ and with $\lambda_i=1$, for every $i$, is $(k+1)$-approximate.
\end{theorem}

\begin{proof}[Proof of Theorem~\ref{thm:hyperstars}] Note that with all $\lambda_i=1,$ we obtain the VCG mechanism on the root players and on the tasks of $T$ in the first summand of the definition, that is, the mechanism giving every task (of $T$) to the player with minimum processing time. It is known that VCG is $k$-approximate:
VCG yields at most the sum of all minimum processing times over all tasks as makespan, and OPT is at least the $1/k$ fraction of this sum. Now, assume that an optimal allocation gives the set  $S^*$
to the root players, and $M\setminus S^*$ to the leaves. Then we have 

\begin{eqnarray*}
ALG  &\leq & \min_{T\subseteq M}\{VCG(T)+\max_{i\not \in  T}\ell_{i}\}\\
    &\leq & VCG(S^*)+\max_{i\not \in  S^*}\ell_{i}\\
    &\leq & k\cdot OPT(S^*)+\max_{i\not \in  S^*}\ell_{i}\\
&\leq & (k+1)\cdot\max\{ OPT(S^*),\max_{i\not \in  S^*}\ell_{i}\}\\
&=& (k+1) OPT,
\end{eqnarray*}
where $VCG(S^*)$ and $OPT(S^*)$ refer only to the root players.
\end{proof}

%% file: objectives.tex
\section{Mechanisms for $L^p$-norm optimization}
\label{sec:lp-norm}

In this section we generalize some of the results of
Section~\ref{sec:scheduling} where the objective is to minimize $L^p$-norm of the
values of the agents, i.e.,
\begin{align}
  \min_X \bigg(\sum_{i=1}^n t_i(X)^p\bigg)^{1/p}.
\end{align}
The makespan scheduling problem is the special case of $p=\infty$. We consider all
positive values of $p$, but we deal separately with the case $p\geq 1$, in which the
$L^p$ is a proper norm, and the case $p\in(0,1)$, where the $L^p$ function is not
subadditive (i.e., the triangle inequality does not hold). We also consider the \emph{maximization} case, which for $p=1$
corresponds to auctions.
 Some of the proofs in this section are more general than the corresponding proofs
in Section~\ref{sec:scheduling}, but similar. Nevertheless, we provide
them for completeness.

\subsection{The Hybrid Mechanism for Minimizing the $L^p$-Norm Objective}
\label{sec:stars-lp}

Consider an instance of the Unrelated Graph Balancing problem on a
star of $n$ nodes and set of tasks $M$. Notice that for stars the
objective of minimizing the $L^p$-norm corresponds to minimizing
$(r(T)^p+\sum_{i\not\in T} \ell_i^p)^\frac{1}{p}$ over all task sets
$T\subseteq M$ given to the $r$-player.

\begin{definition}[Hybrid $L^p$ Mechanism for stars]
  For a given $0< p\leq \infty,$ and an instance of the Unrelated Graph Balancing
  problem on a star of $n$ nodes and set $\,M\,$ of tasks, let
  \begin{align} \label{eq:max2}
    S\in \argmin_{T\subseteq M}\bigg\{r(T)+\bigg(\sum_{i\not \in  T} \ell^p_i\bigg)^{1/p}\bigg\}.
  \end{align}
  The mechanism assigns $S$ to the root and the remaining tasks to leaves. Ties are
  broken in a deterministic way (e.g., lexicographically).
\end{definition}

The argmin expression that defines the Hybrid $L^p$ Mechanism coincides with the VCG
mechanism for $p=1$ and with the Hybrid Mechanism of Section~\ref{sec:scheduling} for
$p\rightarrow\infty$. As it is shown in Section~\ref{sec:hypergraphs} (Corollary~\ref{cor:sufficient}), the Hybrid $L^p$ mechanism is truthful. The reason that leaf players have no incentive to lie comes from the fact that the expression in \eqref{eq:max2} is monotone in $\ell_i$. And the reason that the root  player has no incentive to lie comes from interpreting $-(\sum_{i\not \in  T} \ell^p_i)^{1/p}$ as its payments.

\begin{lemma}
The Hybrid $L^p$ Mechanism for stars is truthful.
\end{lemma}

We now consider the approximation ratio achieved by the mechanism. We summarize here the inequalities that we will use:

\begin{lemma}\label{lem:jensen} For any $p\geq 1$ it holds
  \begin{align} \label{eq:ineqs}
    \sum_{i=1}^k x_i^p \leq \bigg(\sum_{i=1}^k {x_i}\bigg)^p \leq k^{p-1}\sum_{i=1}^k
    x_i^p.
  \end{align}
  Similarly for any $0<p\leq 1$ it holds
  \begin{align}
    \sum_{i=1}^k x_i^{1/p} \leq \bigg(\sum_{i=1}^k {x_i}\bigg)^{1/p} \leq k^{1/p-1}\sum_{i=1}^k
    x_i^{1/p}.
  \end{align} 
\end{lemma}

\begin{proof}
  The left inequality of~\eqref{eq:ineqs} is essentially the triangle inequality of the $L^p$-norm. The
  right inequality is an immediate application of Jensen's inequality for a convex function $\varphi$ 
$$\varphi\left(\frac{\sum_{i=1}^k{x_i}}{k}\right)\leq \frac{\sum_{i=1}^k{\varphi(x_i)}}{k},$$
where $\varphi(x)=x^p$.
The second set of inequalities come by replacing $p$ with $1/p \geq 1$. 
\end{proof}

Next we show two upper bound results for the approximation ratio (for the $L^p$-norm objective) separately in case $p\geq 1,$ and in case $0< p\leq 1,$ respectively.

\begin{theorem} \label{thm:upper-bound-min-lp}
For the problem of minimizing the $L^p$-norm, the Hybrid $L^p$ Mechanism for stars
has approximation ratio of at most $2^{(p-1)/p}$, when $p\geq 1$, and  $2^{(1-p)/p}$,
when $0<p<1$.
\end{theorem}

\begin{proof} Let $S^*=\argmin_{T\subseteq M}(r(T)^p+\sum_{i\not\in T}
  \ell_i^p)^\frac{1}{p}$ be the subset assigned to the root in the optimal
  allocation, $S$ be the subset assigned to the root by the $L^p$ Mechanism, $OPT$ be
  the optimal $L^p$-norm, and $ALG$ be the  $L^p$-norm achieved by the Hybrid $L^p$
  Mechanism.

  We first consider the case $p\geq 1$. We have

\begin{align*}
  ALG  = \bigg(r(S)^p+\sum_{i\not \in  S} \ell^p_i\bigg)^{1/p}
         &\leq r(S)+\bigg(\sum_{i\not \in  S} \ell^p_i\bigg)^{1/p} \\
       &\leq r(S^*)+\bigg(\sum_{i\not \in  S^*} \ell^p_i\bigg)^{1/p} \\
        & \leq  \, 2^{(p-1)/p}\bigg(r(S^*)^p+\sum_{i\not \in  S^*} \ell^p_i\bigg)^{1/p}
= 2^{(p-1)/p} OPT, 
\end{align*}
where the first inequality follows from the triangle
inequality, %
the second from the definition of the Hybrid $L^p$ Mechanism, while the last one from
Jensen's inequality (Lemma~\ref{lem:jensen}, Equation~\eqref{eq:ineqs}) for
$k=2,\, x_1=r(S^*),$ and $x_2=(\sum_{i\not \in S^*} \ell^p_i)^{1/p}.$

The case of $p<1$, is essentially the same, but the proof is slightly different. We
first apply Jensen's inequality and then the ``triangle inequality''.
\begin{align*}
 ALG  = \bigg(r(S)^p+\sum_{i\not \in  S} \ell^p_i\bigg)^{1/p}
&\leq 2^{\frac{1}{p}-1}\bigg(r(S)+\big(\sum_{i\not \in  S} \ell^p_i\big)^{1/p}\bigg) \\
&\leq  2^{\frac{1}{p}-1}\bigg(r(S^*)+\big(\sum_{i\not \in  S^*} \ell^p_i\big)^{1/p}\bigg) \\
&\leq   2^{\frac{1}{p}-1}\bigg(r(S^*)^p+\sum_{i\not \in  S^*} \ell^p_i\bigg)^{1/p}
=\,  2^{\frac{1}{p}-1} OPT. 
\end{align*}
The first inequality follows from Jensen's inequality (Lemma~\ref{lem:jensen}) for
$x_1=(r(S))^p,$ and $x_2=\sum_{i\not \in S} \ell^p_i;$ the second from the definition
of the $L^p$ Mechanism, while the last one from the fact that
$(\alpha+\beta)^p\leq \alpha^p+\beta^p$, when $0<p\leq 1.$
\end{proof}

As in the case of makespan, we can use the mechanism to other domains by decomposing
them. We can apply the Star-Cover mechanism (Definition~\ref{def:star-cover}) to get
good approximation ratios for general domains.

\begin{theorem} For $p\geq 1,$ the Star-Cover mechanism for a given multigraph $G$
  that uses the Hybrid $L^p$ Mechanism on every star of a fixed star decomposition
  $T=\{T_1,\ldots,T_k\}$ is truthful and has an approximation ratio at most
  $(2c(T))^{(p-1)/p}$ of the $L^p$-norm of the machines' costs, where $c(T)$ is the
  star contention number of the decomposition.
\end{theorem}

\begin{proof}
  Fix some player $i$ and let $S_{i,h}$ be the subset of tasks allocated to player $i$
  by the $L^p$ Mechanism when applied to a star $T_h$, $h=1,\ldots,k$. Like in the case of Max Mechanism (Definition~\ref{def:maxmech}), truthfulness follows from two observations. First, since the fixed
  star decomposition is independent of player $i$'s processing times, player $i$
  cannot affect it by lying. Second, $S_{i,h}$ is independent of player $i$'s
  processing times $t_{i}(e)$ for all edges $e\not\in T_{h}$, therefore player $i$
  cannot alter the assignment on $T_h$ by changing its values outside $T_h$.

  Now we show the approximation guarantee. Let $OPT$, $OPT(T_h)$ be the optimal $L^p$-norm
  on $G$ and $T_h$ respectively, and let $ALG$ and $ALG(T_h)$ be the $L^p$-norm
  achieved by the $L^p$-Cover mechanism on $G$ and $T_h.$ Let $c=c(T)$ be the star contention number of $T.$ We prove that $(ALG)^p\leq (2c)^{p-1} (OPT)^p.$
  
  \begin{align*}
   (ALG)^p  &= \sum_i \bigg(\sum_h t(S_{i,h})\bigg)^p\\ 
   & \leq \sum_i c^{p-1}\sum_h t(S_{i,h})^p = c^{p-1}\sum_h \sum_i t(S_{i,h})^p = c^{p-1}\sum_h (ALG(T_h))^p \\
&\leq  c^{p-1}\sum_h 2^{p-1}(OPT(T_h))^p =  (2c)^{p-1}\sum_h (OPT(T_h))^p =  (2c)^{p-1}\sum_h \sum_i (t(\tilde S_{i,h}))^p \\
&\leq   (2c)^{p-1}\sum_h \sum_i (t(S^*_{i,h}))^p =   (2c)^{p-1}\sum_i \sum_h  (t(S^*_{i,h}))^p \\
&\leq   (2c)^{p-1}\sum_i  (t(S^*_{i}))^p =  (2c)^{p-1} ( OPT)^p. 
  \end{align*}
The  first inequality  follows from Lemma~\ref{lem:jensen}, because for every machine
$t(S_{i,h})$ is nonzero only for at most $c$ stars $h\in \{1,2,\ldots k\}.$ The
second holds by the approximation ratio of the Hybrid $L^p$ Mechanism for stars. $\tilde S_{i,h}$ denotes the set given to machine $i$ in the optimal allocation \emph{of star} $h,$ whereas $S^*_{i,h}$ is the restriction of the allocated task set to player $i$ by $OPT,$ to the tasks of star $h.$ The third inequality holds by the optimality of the $\tilde S_{i,h}$ on each star $h.$ Finally, the last inequality holds by the triangle inequality $\sum x_i^p\leq (\sum x_i)^p.$\end{proof}

\subsection{Lower Bounds for Minimizing the $L^p$-norm}
\label{sec:lower-bounds-lp}

We now provide corresponding negative results for mechanisms. For the case, of
$p\geq 1$, the next theorem shows that the Hybrid $L^p$ Mechanism has optimal
approximation ratio. The case of $p<1$ is treated separately below, and the lower
bound that we give does not match exactly the upper bound, which leaves open the
possibility that there exists a mechanism with better approximation ratio than the
Hybrid $L^p$ Mechanism.

\begin{theorem} \label{thm:lower-bound-stars-lp} For any $p\geq 1$,
  there is no deterministic mechanism for stars that can achieve an
  approximation ratio better than $2^{1-1/p}$ for the $L^p$-objective.
\end{theorem}

\begin{proof}
  This is a generalization of the corresponding proof in
  Section~\ref{sec:scheduling}.
  
  Let's assume that the mechanism takes an input where the processing
  times of the root player are $r_j =2^{j-1}$, for each task $j$, and the
  processing time of the corresponding leaf player for task $j$ is
  $\ell_j=\alpha\cdot 2^{j-1}$, for some fixed $\alpha>0$ to be determined. 

If the mechanism assigns all tasks to the root player, then the
makespan for this input is 
$$\sum_{j=1}^mr_j=2^m-1,$$ while the optimal makespan is
$$\bigg(\sum_{j=1}^{m-1}\ell^p_j+r^p_m\bigg)^{1/p}= \left(\frac{\alpha^p(2^{p(m-1)}-1)+2^{p(m-1)}(2^p-1)}{2^p-1}\right)^{1/p}.$$ 

Therefore the approximation ratio is 
$$\left(\frac{(2^m-1)^p(2^p-1)}{\alpha^p(2^{p(m-1)}-1)+2^{p(m-1)}(2^p-1)}\right)^{1/p},$$

 which tends to the following value as $m$ tends to infinity 

\begin{equation}
\label{eq:1}\left({{\frac {{2}^{p} \left( {2}^{p}-1 \right) }{{a}^{p}+{2}^{p}-1
}}}\right)^{1/p}.
\end{equation}

Otherwise, let $X$ be the nonempty set of tasks assigned to the leaf
players. Let $k$ be the task with the maximum index in $X$. Since it
is processed by the leaf player, its processing time is $\ell_k=\alpha\cdot 2^{k-1}$. Now
consider the input in which we change the processing times of the root
player to

$$r'_j=
\begin{cases}
  0 & j\not\in X \\
  r_j+\epsilon & \text{ otherwise} 
\end{cases}
$$
for some arbitrarily small $\epsilon>0$. By weak monotonicity
(Lemma~\ref{lemma:tool}), the set of tasks assigned to the root player remains the
same, and as a result the whole allocation stays the same. Therefore task $k$ is
still assigned to the leaf player $k$ and the makespan of the mechanism is 
$\left(\sum_{j\in X}\ell_j^p\right)^{1/p}$, while the optimum
allocation for this input is at most $\left(\sum_{j\in
    X\setminus\{k\}}\ell_j^p+r_k^p\right)^{1/p}$.

Therefore the approximation ratio is at least

$$\frac{\left(\sum_{j\in X}\ell_j^p\right)^{1/p}}{\left(\sum_{j\in
    X\setminus\{k\}}\ell_j^p+r_k^p\right)^{1/p}}\geq \frac{\left(\sum_{j=1}^k\ell_j^p\right)^{1/p}}{\left(\sum_{j=1}^{k-1}\ell_j^p+r_k^p\right)^{1/p}}=\left({\frac {{a}^{p} \left( {2}^{p \left( k+1 \right) }-1 \right) }{{a}^{p}
 \left( {2}^{pk}-1 \right) +{2}^{pk} \left( {2}^{p}-1 \right) }}
\right)^{1/p}
$$

For large $k$ this tends to 

\begin{equation}
\label{eq:2}
\frac{2\alpha}{\left(\alpha^p+2^p-1\right)^{1/p}}.
\end{equation} 

By setting $\alpha=(2^{p}-1)^{1/p}$ both (\ref{eq:1}) and (\ref{eq:2}) become equal to $2^{1-1/p}$ and the theorem follows.

\end{proof}

Before we proceed to give a lower bound for the case of $p<1$, we point out that all
known mechanisms perform much worse than the Hybrid Mechanism.

\begin{theorem} \label{thm:local-norm} For minimizing the $L^p$-norm on stars, all local
  mechanisms, including affine minimizers and task-independent mechanisms, have
  approximation ratio of at least
  $m^{\frac{1}{2}(1-1/p)}=(n-1)^{\frac{1}{2}(1-1/p)},$
  when $p\geq 1$.
\end{theorem}

Observe that for $p=1$, the VCG is optimal, but for large $p$ the inefficiency of all
local mechanisms grows and tends to $\sqrt{m}$. This is essentially the same with the
lower bound of the corresponding theorem in Section~\ref{sec:scheduling} and we omit
its proof.

We now give a lower bound for all mechanisms for the case of $p<1$. Notice that the
approximation ratio tends to infinity as $p$ tends to 0.

\begin{theorem} \label{thm:lower-bound-stars-lp} For any $0<p\leq 1$ and every
  $a> 1$, there is no deterministic mechanism for stars that can achieve an
  approximation ratio better than
  \begin{align}
    \min \bigg\{a, \frac{(a+1)^{1/p}}{a^{1/p}+a)}\bigg\}.
  \end{align}
  By selecting an appropriate $a$, this is $\Omega(p^{-1} / \ln(p^{-1}))$.
\end{theorem}

\begin{proof}
  Consider an instance with two tasks, where the root has costs
  $r_1=a^{1/p}, r_2=a$, and the leaves have costs
  $\ell_1=\infty,\ell_2=1$.  Observe that in this instance,
  the optimum is to assign both tasks to the root player, with a total
  cost of $a^{1/p}+a$. Any other allocation has total cost at least
  $(a+1)^{1/p}$, which gives an approximation ratio at least
  $(a+1)^{1/p}/(a^{1/p}+a)$.

  However, if the mechanism assigns both tasks to the root, then consider the instance
  produced by setting $r'_1=0,r'_2=a-\epsilon$. By applying the monotonicity lemma
  (Lemma~\ref{lemma:tool}), the allocation remains the same, with a cost of
  $a-\epsilon$, while the optimum is 1. Letting $\epsilon$ go to 0, we get a ratio
  $a$.

  To show that this bound is almost proportional to $1/p$, we show that for
  $a>1$ and $q= 1/p \geq e^2$, we have
  \begin{align*}
    \min\{(a+1)^q/(a^q+a), a\} = \Omega(q / \ln q ).
  \end{align*}
  Indeed we have
  \begin{align*}
    (a+1)^q/(a^q+a) & \geq (a+1)^q/(2a^q) \\
                            &\approx \frac{1}{2} e^{q/a}, 
  \end{align*}
  If $a \leq q / \ln q$, the above is at least
  $\frac{1}{2} e^{q/a} \geq \frac{1}{2} q \geq q / \ln q$.
  \end{proof}

  The lower bound given by the above theorem does not match the upper bound of
  Theorem~\ref{thm:upper-bound-min-lp}. For example, for $p=1/2$, the theorem above
  gives a lower bound of $\phi\approx 1.618$, while the upper bound is equal to 2.
  
  \subsection{Maximizing the $L^p$-norm  (Auctions)}

  In this subsection, we illustrate that the Hybrid Mechanism has good performance
  for the maximization problem as well.
  
\begin{definition}[Hybrid Mechanism for the $L^p$-Norm (maximization version)]
  Consider an instance of the Unrelated Graph Balancing problem on a star of $n$
  nodes and set of tasks $M$. Let
  \begin{align} \label{eq:max3}
    S\in \arg\max_{T\subseteq M}\bigg\{r(T)+\bigg(\sum_{i\not \in  T} \ell^p_i\bigg)^{1/p}\bigg\}.
  \end{align}
The mechanism assigns $S$ to the root and the remaining tasks to leaves. Ties are
broken in a deterministic way (e.g., lexicographically).
\end{definition}

Note that for $p=1$, the Hybrid Mechanism coincides with the VCG mechanism.

\begin{lemma}
  The Hybrid Mechanism for maximizing the $L^p$-norm on stars has approximation ratio
  of at most $2^{(p-1)/p}$, when $p\geq 1$.
\end{lemma}

We omit the proof, which is similar to the proof for the minimization version and is
based on applying the triangle and Jensen inequality.

However, unlike the minimization version, there is an even better mechanism for the
star, when $p\geq 2$.

\begin{definition}[All-or-Nothing Mechanism]
  Consider an instance of the Unrelated Graph Balancing problem on a star of $n$
  nodes and set of tasks $M$. If
  \begin{align} \label{eq:max4}
    r(M) \geq \bigg(\sum_{i\in M} \ell^p_i\bigg)^{1/p}.
  \end{align}
then the mechanism assigns $M$ to the root, otherwise it assigns all tasks to the leaves.
\end{definition}

\begin{lemma}
The All-or-Nothing Mechanism for the star is truthful and has approximation ratio of at most $2^{1/p}$.
\end{lemma}

\begin{proof}
  Let $S^*$ be the subset assigned to the root in the optimal allocation, and $OPT$
  be the optimal makespan, and $ALG$ be the makespan achieved by the All-or-Nothing
  Mechanism. Then we have

\begin{align*}
  OPT & = \bigg(r(S^*)^p+\sum_{i\not \in  S^*} \ell^p_i\bigg)^{1/p}
        \leq \bigg(2\max\bigg\{r(S^*)^p,\sum_{i\not\in S^*} \ell^p_i\bigg\}\bigg)^{1/p} 
        \leq \bigg(2\max\bigg\{r(M)^p,\sum_{i\in M} \ell^p_i\bigg\}\bigg)^{1/p}
        = 2^{1/p}ALG. 
\end{align*}
\end{proof}

If we use the Hybrid Mechanism for $1\leq p\leq 2$ and the All-or-Nothing Mechanism
for $p>2$, we get a mechanism with approximation ratio
$\min\{2^{1/p},2^{(p-1)/p}\}\leq \sqrt{2}$. We don't know whether this bound is
tight, but below we show a weaker lower bound for the case of $p=2$. 

\begin{theorem} \label{thm:lower-bound-stars-lp-max} For $p=2$,
  there is no deterministic mechanism for stars that can achieve an
  approximation ratio better than $1.05$.
\end{theorem}

\begin{proof}
  For the sake of exposition, we use simple values. Optimizing them could lead to a
  slightly higher ratio.

  Consider an instance with two tasks, where the root has costs $r_1=1, r_2=1/3$, and
  the leaves have costs $\ell_1=0,\ell_2=1$.  Observe that in this instance, the
  optimum is to assign the first task to the root and the second to the leaf, for a
  total cost of $\sqrt{2}$. If the mechanism assigns both tasks to the root, its cost
  is $4/3$ and the ratio is higher than $1.05$.

  Therefore it must be the case that the root takes only the first task.  In that
  case, consider the instance produced by setting $r'_1=3,\, r'_2=1/3-\epsilon$. By
  applying the weak monotonicity lemma (Lemma~\ref{lemma:tool}), the allocation
  remains the same, with a cost of $\sqrt{10}$. The optimum is $10/3 - \epsilon$,
  which yields a ratio that is higher than $1.05$, when $\epsilon$ goes to 0.
  \end{proof}

%% file: hybrid.tex
\section{Hybrid Mechanisms}
\label{sec:hybrid}
Here we provide the general definitions related to Hybrid Mechanisms,
and show necessary and sufficient conditions for truthfulness on stars
and hyper-stars. We emphasize that this is a
multi-dimensional mechanism design setting. Each leaf $j$ has a single
dimensional valuation, given by the scalar $\ell_j$ but a root has
multi-dimensional preferences, given by the vector of
values. %

For the sake of convenience, we call non-decreasing real functions
\emph{increasing}, and non-increasing functions \emph{decreasing}. We
say \emph{strictly increasing/decreasing} if we want to emphasize
strict monotonicity.

It is known, that an allocation rule can be equipped by a truthful
payment scheme iff it is \emph{weakly monotone} (definition~\ref{def:wmon}). The next
two propositions give a characterization of the weak monotonicity
property in our case, for the leaf-players, and for the root player (in case of stars),
respectively:

\begin{proposition}
\label{prop:leafwmon} An allocation rule is weakly monotone for a leaf-player $i,$  iff for every $r$ and every $\ell_{-i},$ whenever leaf-player $i$ gets task $i$ with bid $\ell_i,$ then he also gets the task with every smaller bid $\ell_i'<\ell_i.$ 
\end{proposition}

\begin{proposition}
\label{prop:rootwmon} An allocation rule on stars is weakly monotone for the root player if and only if for every fixed bid vector $\ell$ of the other players, and every  $T\subseteq M$ a constant  $g_T(\ell)$ (i.e., independent of $r$) exists, such that for every $r$ the root player is allocated a  set  $S\in \arg\min_T\,\{r(T)+g_{T}(\ell)\}.$ 

The canonical choice for trutful payments to the $r$-player is then  $P^0_S(\ell)=g_\emptyset(\ell)-g_S(\ell),$ and all other truthful payments can be obtained by an additive shift by an arbitrary $c(\ell).$  
\end{proposition}

Since the $r$-player (player $0$) wants to maximize his profit, and the costs for the tasks are nonnegative, it will be convenient to assume w.l.o.g. that for every fixed  $\ell$ the payments $P^0_S$ correspond to an increasing set-function of $S,$\footnote{We call a setfunction $P$ \emph{increasing}, if $P(S')\leq P(S)$  whenever $S'\subset S,$ and \emph{strictly increasing} if the inequality is strict.} because a set of tasks with higher cost and less payments can not  be allocated to him by a truthful mechanism.\footnote{See also the \emph{virtual payments} in \cite{CKK20}.}

Motivated by Proposition~\ref{prop:rootwmon} we restrict our search for truthful mechanisms on star graphs as follows:

\begin{definition}\label{def:star}[Hybrid Mechanism]
 Assume that an $m$-variate function $g_T:\mathbb{R}^{m}\rightarrow \mathbb{R}$ is given for every $T\subseteq M,$ so that for every fixed vector $\ell\geq 0$ the values $\{g_T(\ell)\}_{T\subseteq M}$ correspond to a decreasing setfunction of $T.$ 
 For any input $(r,\ell),$ a \emph{Hybrid Mechanism (for the functions $\{g_T\}_{T\subseteq M}$ )} allocates a set $S$ to the root player such that 
 
$$S\in \arg\min_T\,\{r(T)+g_{T}(\ell)\};$$
if there are more than one such sets $S,$ the mechanism breaks ties according to the lexikographic order over all subsets of $M.$ 
The items in $M\setminus S$ are assiged to the leaves.

\end{definition}

The more general definition for hyperstars is Definition~\ref{def:hybrid}. For the sake of simplicity, the next two lemmas are formulated for star graphs. In the proof of the subsequent main lemma (of Lemma~\ref{lem:charstar}) we discuss the necessary changes for hyperstars.

Consider a Hybrid Mechanism on a star. For any $i\in M$ fix all bids in the input except for $r_i,$ i.e., fix the vectors $r_{-i}$ and $\ell.$ 

\begin{definition}\label{def:critvalue} The following  function defines the so called \emph{critical value} for the bid $r_i$:

$$\psi_i=\psi_i[r_{-i},\ell]=\min_{T:i\notin T}\{r(T)+g_T(\ell)\}-\min_{T:i\in T}\{r(T\setminus \{i\})+g_T(\ell)\}$$
We omit the argument $r_{-i},\ell$ whenever they are obvious from the context.

\end{definition}

\begin{lemma}\label{lem:psiinc} $\psi_i[r_{-i},\ell]\geq 0$ for every $r_{-i},\ell.$
\end{lemma}

\begin{proof} The proof follows from the fact that for every fixed $\ell$ the function $g_T(\ell)$  is a decreasing setfunction of $T.$ %
This implies that if $i\notin T,$ then $g_T(\ell)\geq g_{T\cup \{i\}}(\ell),$ and so
$r(T)+g_T(\ell)\geq r(T)+g_{T\cup \{i\}}(\ell).$ The same inequality must therefore hold for the minimium values 

$$\min_{T:i\notin T}\{r(T)+g_T(\ell)\}\geq \min_{T:i\notin T}\{r(T)+g_{T\cup \{i\}}(\ell)\}.$$ Observe that first and the second expressions in the definition of $\psi_i$ correspond precisely these two expressions,  which concludes the proof.
\end{proof}

We show next that $\psi_i$ is, indeed, a critical value function:

\begin{lemma}
\label{lem:critvalue}
Let $i\in M,$ and arbitrary nonnegative bid vectors $ \,r_{-i}$ and $\,\ell$ be fixed; then for every $r_i<\psi_i$ the root player receives task $i,$ and for every $r_i>\psi_i$ the leaf player with bid $\ell_i$ receives task $i.$
\end{lemma}
\begin{proof} If $r_i<\psi_i,$ then 
$$r_i+\min_{T:i\in T}\{r(T\setminus \{i\})+g_T(\ell)\}<\min_{T:i\notin T}\{r(T)+g_T(\ell)\}$$
or equivalently
$$\min_{T:i\in T}\{r(T)+g_T(\ell)\}<\min_{T:i\notin T}\{r(T)+g_T(\ell)\}.$$
Therefore, no set $S$ for which $i\notin S,$  can be in $\arg\min_T \{r(T)+g_T(\ell)\},$
and since the Hybrid Mechanism minimizes the expression $r(T)+g_T(\ell),$ a set $S\subseteq M$ with $i\in S$ will be allocated to the root player.
On the other hand, if  $r_i>\psi_i,$ then all inequalities in the above argument get flipped, and a set with $i\notin S$ will be given to the root player. 
\end{proof}

The following lemma provides various necessary or sufficient conditions for the truthfulness of Hybrid Mechanisms in stars in terms of monotonicity of the critical value function $\psi_i$ as a function of $\ell_i.$ For hyperstars we (can) only give sufficient conditions. %

\begin{lemma}
\label{lem:charstar}
For the truthfulness of the Hybrid Mechanism with given $\{g_T\}_{T\subseteq M}$ functions (i.e., for a truthful payment scheme to exist), 

\begin{itemize}
\item[(a)] in stars it is necessary that for every $i\in M$ and every fixed $(r_{-i},\ell_{-i})$ the function $\psi_i(\ell_i)=\psi_i[r_{-i},\ell_{-i}](\ell_i)$ is an increasing function of $\ell_i;$

\item[(b)] in (hyper-)stars it is sufficient that for every $i\in M$ and every fixed $(r_{-i},\ell_{-i})$ the function $\psi_i(\ell_i)=\psi_i[r_{-i},\ell_{-i}](\ell_i)$ is a \emph{strictly} increasing function of $\ell_i;$

\item[(c)] in (hyper-)stars it is sufficient that for every $i$ and $\ell_{-i}$ the  $g_T(\ell_i,\ell_{-i})$ is an increasing function of $\ell_i$ whenever $i\notin T,$ and decreasing function of $\ell_i$  whenever $i\in T.$ 
\end{itemize}
\end{lemma}
\begin{proof} We prove the lemma for stars first. By Proposition~\ref{prop:rootwmon}, an allocation rule is weakly monotone for the root player of a mechanism iff it is a minimizer with the additive values $g_T(\ell).$ Furthermore it is without loss of generality to admit only $g_T(\ell)$ so that for fixed $\ell$ the payments $P^0_T(\ell)=g_\emptyset(\ell)-g_T(\ell),$ make a (normalized) increasing setfunction of $T $~     \cite{CKK20}. This is the case iff the $g_T(\ell)$ is a decreasing setfunction of $T,$ altogether iff the mechanism is a Hybrid Mechanism by Definition~\ref{def:star}. The conditions (a)--(c) are necessary or sufficient for weak monotonicity for the \emph{leaf} players, as we show next.

For (a) we need to prove that the monotonicity of the $\psi_i$ functions is necessary. Assume that there exist$\,i, r_{-i},$ and $\ell_{-i}$ so that $\psi_i$ is not monotone increasing, i.e., there are $\ell'_i<\ell_i$ values so that  $\psi_i(\ell'_i)>\psi_i(\ell_i).$ We take an $r_i$ between these values, i.e., let $\psi_i(\ell_i)<r_i< \psi_i(\ell'_i).$ 
Now consider the leaf-player $i$ for the bids $(r_i,r_{-i}),$ and $\ell_{-i}$ of all other players. According to Lemma~\ref{lem:critvalue}, with bid $\ell_i$ this leaf-player gets task $i$ because $\psi_i(\ell_i)<r_i,$ but with the smaller bid $\ell_i'$ he does not get task $i$ because $r_i< \psi_i(\ell_i').$ Thus the mechansism cannot be truthful for the leaf-player $i$ by Proposition~\ref{prop:leafwmon}.

 Next we prove that (b) is sufficient for weak monotonicity of the allocation to an arbitrary leaf player $i.$ Assume  that a mechanism with strictly increasing $\psi_i$ functions is not weakly monotone, then by Proposition~\ref{prop:leafwmon} there exist $r, \ell_{-i},$  and  $\ell'_i<\ell_i$ so that player $i$ gets task $i$ with bid $\ell_i,$ but does not get it with bid $\ell_i'.$ Then, by Lemma~\ref{lem:critvalue} $r_i$ must be such that $\psi_i[r_{-i},\ell_i,\ell_{-i}]\leq r_i\leq \psi_i[r_{-i},\ell_i',\ell_{-i}].$ Since  $\psi_i$ is a strictly increasing function of $\ell_i,$ it must be the case that  $\psi_i(\ell_i)>\psi_i(\ell_i'),$ a contradiction.

 In order to show that (c) is sufficient, let again $i$ be  a leaf-player, and assume for contradiction that there exist $r, \ell_{-i},$  and  $\ell'_i<\ell_i''$ so that leaf player $i$ gets task $i$ with bid $\ell_i'',$ but does not get it with bid $\ell'_i.$
 Then $r_i$ must be such that $\psi_i[r_{-i},\ell_i'',\ell_{-i}]\leq r_i\leq \psi_i[r_{-i},\ell_i',\ell_{-i}].$ Since  $\psi_i$ is increasing in $\ell_i$ (by the assumptions of (c)), it must be the case that $\psi_i(\ell_i'')=\psi_i(\ell_i').$ Let $T^*$ be the set with highest priority among sets with $i\in T$ that minimize the expression $r(T)+g_T(\ell_i'').$ %

  Since the leaf player gets task $i$ for bid $\ell_i'',$ there must be a set $S$ with $i\not\in S$ and 
 so that $r(S)+g_{S}(\ell_i'')\leq r(T^*)+g_{T^*}(\ell_i'')$ (with higher priority than $T^*$ in case equality holds). Given that $g_S$ is increasing  and $g_{T^*}$ decreasing function of $\ell_i,$ this $S$ must beat $T^*$  for $\ell_i'$ as well. Assume that another $T^{**}$ would beat $S$ at $\ell_i',$ then this  $T^{**}$  would necessarily beat $S$ and $T^{*}$ also at $\ell_i'',$  a contradiction. This concludes the proof for stars.%
 
 We discuss the necessary changes  in the above proof for hyperstars. We need to prove weak-monotonicity for an arbitrary root-player, and for an arbitrary leaf-player in cases (b) and (c).
 For a root-player $h,$ weak monotonicity holds, if (but not only if) we replace $g_T(\ell)$ in Proposition~\ref{prop:rootwmon} by  $$\tilde g_T(\ell,r_{-h}):=\min_{R\subseteq M\setminus T}\{\min_{x^R}\sum_{j\leq k\, | \,j\neq h}\lambda_j r_j\cdot x^R_j  +g_{R\cup T}(\ell)\}.$$ Here, $r$ is not a vector but a $k\times m$ matrix, and the notation $\,r_{-h}\,$ refers to the bids of the root-players other than player $h.$  Note that  this expression corresponds again  to the minimum sum in the allocation rule that all players except for $h$ can have, given that player $h$ receives $T.$
As required, the $\tilde g_T(\ell,r_{-h})$ are increasing set-functions of $M\setminus T$ for fixed $\ell$ and $r_{-h}:$ assume that $M\setminus T'\subset M\setminus T,$ and let $R\subset M\setminus T$ provide the minimum for $\tilde g_T.$ Then $R\cap(M\setminus T')$ yields a not larger value over $M\setminus T'$ which, in turn, is an upper bound for $\tilde g_{T'}.$ (Here we exploit that $g$ is an increasing set-function.) 

For the leaf-players, in Definition~\ref{def:critvalue},
Lemmas~\ref{lem:psiinc} and \ref{lem:critvalue} we need to replace $r(T)$ by
$\min_{x^T}\sum_{j=1}^k\lambda_j r_j\cdot x^T_j$ for hyperstars. In general, this is
the smallest sum that the root players can achieve on the tasks in $T.$ If now
$\min_{1\leq j\leq k} \lambda_j\cdot r_{ji}$ (instead of $r_i$ in the single-root case) is below the critical value $\psi_i$
then the corresponding root player gets task $i,$ otherwise the leaf-player $k+i.$
The rest of the proof of (b) and (c) is analogous%
.
 \end{proof}

The following example shows that conditions  (b) and (c) are both \emph{not} necessary for the Hybrid Mechanism to be truthful.

\begin{example} Let $m=2,$ and consider the following functions:  $g_{\{1,2\}}=0,\,\, g_{\{1\}}=\ell_2+1/\ell_1,\,\,g_{\{2\}}=1,\,\,g_\emptyset=\ell_2+1/\ell_1+1,\,$ that is, we minimize over $$r_1+r_2,\quad r_1+\ell_2+1/\ell_1,\quad r_2+1,\quad \ell_2+1/\ell_1+1.\,$$
For every fixed $\ell,$ $g_T(\ell)$ is a decreasing  setfunction of $T$ (in fact, additive function of $M\setminus T$). The critical value functions are $\psi_1\equiv 1$ for every $t_2$ and every $\ell,$  and $\psi_2=\ell_2+1/\ell_1$ for every $t_1.$

Observe that  $\psi_1$ is not strictly increasing, and $g_\emptyset=1+\ell_2+1/\ell_1$ is not increasing in $\ell_1,$ although $1\notin \emptyset.$
Thus, this mechanism fulfils neither (b) nor (c).

Still, for example with a tie-breaking that prefers giving the tasks (independently) to the root player, it is truthful.
\end{example}

\begin{corollary}\label{cor:sufficient}
The Hybrid Mechanism for Graph Balancing, the Hybrid Max-Mechanism for Hypergraphs,  and the Hybrid $L^p$ Mechanism on stars are truthful.
\end{corollary}

\begin{proof} The first statement follows from the fact that the Hybrid Mechanism for Graph Balancing fulfils (c).
Clearly, $g_T(\ell)=\max_{i \notin
  T}\{\ell_i\}=\max_{i \in
  M\setminus T}\{\ell_i\}$ is an increasing setfunction of the sets $M\setminus T,$
  and therefore a decreasing setfunction of the sets $T,$ for fixed $\ell.$
  For fixed $T,$  $\max_{i \notin
  T}\{\ell_i\}$ is an increasing function of $\ell_i$ for every $i\notin T,$ and it is independent of $\ell_i$ (constant function) if $i\in T.$
  
  Finally, it is easy to see that the Hybrid $L^p$ Mechanism of Section~\ref{sec:lp-norm} fulfils (b) as well as (c).
  \end{proof}